\documentclass[conference]{IEEEtran}
\IEEEoverridecommandlockouts

\usepackage{graphics}
\usepackage{amsmath}
\usepackage{amssymb}
\usepackage{epsfig,wrapfig}
\usepackage{graphicx}
\usepackage{amsmath}
\usepackage{epsfig,wrapfig}
\usepackage{color}
\usepackage{cite}

\newtheorem{definition}{Definition}
\newtheorem{lemma}{Lemma}
\newtheorem{theorem}{Theorem}
\newtheorem{corollary}{Corollary}

\newtheorem{example}{Example}

\title{\LARGE \bf Modal Barriers to Controllability in Networks with Linearly-Coupled Homogeneous Subsystems}
\author{Mengran Xue\thanks{The authors are with the School of Electrical Engineering and Computer Science at Washington State University. This work has been generously supported by National Science Foundation Grants CNS-1545104 and CMMI-1635184.  Correspondence should be sent to sroy@eecs.wsu.edu. This paper has been submitted to the IEEE Transactions on Automatic Control.  Copyright is held by M. Xue and S. Roy.} and Sandip Roy}

\begin{document}
\maketitle

\begin{abstract}
The controllability of networks comprising homogeneous multi-input multi-output linear subsystems with linear couplings among them is
examined, from a modal perspective.  The  eigenvalues of the network model are classified into two groups: 1) network-invariant modes, which have very high multiplicity regardless of the network's topology; and 2) 
special-repeat modes, which are repeated for only particular network topologies and have bounded multiplicity.
Characterizations of both types of modes are obtained, in part by drawing on decentralized-fixed-mode and generalized-eigenvalue concepts.  We demonstrate that network-invariant modes necessarily prevent controllability unless a sufficient fraction of the subsystems are actuated, both in the network as a whole
and in any weakly-connected partition.  In contrast, the multiplicities of special-repeat modes have no influence on controllability.  Our analysis highlights a distinction between built networks where subsystem interfaces may be unavoidable barriers to controllability, and multi-agent systems where protocols can be designed to ensure controllability.
\end{abstract}

\section{Introduction}

Networks models made up of homogeneous subsystems with linear
couplings among them primarily arise in two contexts: 1) the analysis of synchronization phenomena in built networks like the electric power grid or
oscillator circuits \cite{sync1,sync2,sync3}; and 2) the design of control protocols for multi-agent systems
\cite{cons1,cons2,cons3}.  It is recognized that the methods used in the two contexts are very closely connected, with
the earlier work on synchronization often guiding the design of controls for multi-agent systems \cite{chen1}.
In both contexts, one recent focus has been to characterize the controllability of these linearly-coupled networks, when inputs can only be applied at a subset of the subsystems \cite{control1,control2,control3,control4,control5}.  These controllability analyses have specifically been motivated by questions about: 1)  the manipulability of built networks using a sparse set of inputs and 2) the guidance of multi-agent systems from a set
of leader agents. A main outcome of these analyses has been to distill controllability for the full model into global (network-level) and subsystem-level controllability conditions, whereupon purely graph-theoretic
results can be obtained.

The purpose of this technical note is to examine structural barriers to the controllability of linear-coupled networks that arise due to characteristics of these systems' spectra.  
The genesis of the study is an omission in the controllability analysis of linearly-coupled networks \cite{control1,miss1}, which
was identified in our recent work \cite{comment,control2} and also independently in \cite{control3}.  These efforts
show that global and subsystem-level controllability do not necessarily guarantee 
controllability of the full model, in the case  that the coupled subsystems are general
multi-input multi-output (MIMO) devices.  Instead, the global and subsystem models can intertwine in a complex way in deciding controllability, because of subtleties in the eigenvector analysis of the models in the case where the full dynamics of the 
network has repeated eigenvalues.  Based on this 
recognition, here we determine scenarios under which linearly-coupled networks have repeated eigenvalues, and study whether these repeated-eigenvalue scenarios are barriers to controllability.  The main outcome of the study
is a dichotomy of eigenvalues in linearly-coupled
network models with regard to their multiplicity.  In particular, eigenvalues of the model can be distinguished as: 1) {\em network-invariant modes}, which have very high multiplicity and necessarily prevent controllability unless a high fraction of network nodes are actuated; and
2) {\em special-repeat modes} which have bounded
multiplicity, and do not influence controllability.  
Characterizations are obtained for the two classes of repeated eigenvalues, in part by drawing on fixed-mode notions for decentralized systems and generalized-eigenvalue analysis techniques.

Our study also highlights a distinction in
the controllability of built networks as compared to multi-agent systems.  Interestingly, although the omission in
the controllability analysis has arisen in both contexts \cite{control1,miss1}, 
there is an important body of work in the multi-agent-systems literature that avoids the issue, and correctly obtains
graph-theoretic conditions for controllability \cite{control5}.  It does so
 by {\em designing}
the interaction protocols or input-output couplings among the agents (subsystems) to guarantee controllability.  Our study confirms that protocol design
can be used generally to eliminate nontrivial network-invariant modes, and hence to reduce controllability to a purely network-graph-based analysis.

The note is organized as follows.  The network model is described in Section II.  Spectral conditions for controllability are briefly summarized in Section III.  In Section IV, the dichotomy in repeated-eigenvalue scenarios
is detailed, and its implications on controllability 
are characterized from an algebraic and graph-theoretic standpoint.

\section{Model}

A network comprising identical linear subsystems with linear couplings, which is subject to external actuation at a subset of the subsystems, is considered.  Formally, a network with 
$N$ subsystems, labeled $1,\hdots,N$, is considered.  Each subsystem $i$ has a state ${\bf x}_i(t) \in R^n$ which evolves in continuous time ($t \in R^+$).  The state of each subsystem $i$ is governed by:
\begin{equation}
\dot{\bf x}_i = A {\bf x}_i +B(\sum_{j=1}^N g_{ij}
C {\bf x}_j+\alpha_i {\bf u}_i),
\end{equation}
where $A \in R^{n \times n}$, $B \in R^{n \times m}$, and 
$C \in R^{m \times n}$ are the common state, input, and output matrices of each subsystem; the scalars $g_{ij}$ indicate the strengths of the couplings between the subsystems; 
${\cal S}$ is a set of $M$ (distinct) subsystems which can be actuated; ${\bf u}_i \in R^m$ is the input provided to subsystem
$i$; and $\alpha_i=1$ if $i \in {\cal S}$ and $\alpha_i=0$ otherwise.  

The network's dynamics can be written in assembled form as:
\begin{equation}\label{full}
\dot{\bf x} = (I_N \otimes A + G \otimes BC){\bf x}+
(S \otimes B) {\bf u},
\end{equation}
where ${\bf x}=\begin{bmatrix} {\bf x}_1 \\ \vdots \\
{\bf x}_N \end{bmatrix}$, ${\bf u}$ stacks the input vectors
for the subsystems listed in set ${\cal S}$, $G=[g_{ij}]$, the matrix $S \in R^{N \times M}$ has columns which are 0--1 indicators vectors of the subsystems listed in ${\cal S}$, $I_N$ is an identity matrix with $N$ rows and columns, and
$\otimes$ represents the Kronecker product operation.
We notice that the triple $(C,A,B)$ specifies the input-output behavior of each subsystem, hence we refer to it as
the {\em subsystem model}.  Meanwhile, the {\em network matrix} $G$ and the {\em actuation-location matrix}
$S$ specify the interaction weights among subsystems and external input locations, respectively, hence we refer to the pair $(G,S)$ as the 
{\em global model}.  The subsystem model and the global model
together define the full network dynamics \eqref{full}, which
we refer to as the {\em linearly-coupled network}. 

Our focus in this study is on the controllability of the linearly-coupled network, as specified
by the pair $(I_N \otimes A + G \otimes BC, S \otimes B)$.
Many of our analyses are limited to the
case that the network matrix $G$ is diagonalizable, which encompasses the symmetric case.  We note that many of the network models considered in the literature assume that 
$G$ has a diffusive structure \cite{sync1,sync2,sync3,cons1,cons2,cons3,control1,control5}, which means that 
$G$ is an essentially-nonnegative matrix (sometimes with the additional requirement that its row sums are $0$); however, our analyses here do not depend on this additional 
structure.

To develop graph-theoretic results, we also find it convenient to define a weighted digraph $\Gamma$ which specifies the topology
of interactions among the subsystems.  Specifically, $\Gamma$ is defined as a graph on $N$ vertices labeled $1, \hdots, N$, which correspond to the subsystems in the network.  An edge is drawn from vertex $j$ to vertex $i$ ($i \neq j$) if $g_{ij} \neq 0$.   The vertices corresponding
to the actuated subsystems (i.e., the vertices identified
in ${\cal S}$) are referred to as the actuation locations in the graph.

The linearly-coupled network encompasses models for both
built-network dynamics (e.g., power-grid swing
dynamics, oscillator synchronization processes) and 
multi-agent-system coordination.  In built-network applications, all network parameters including the interfaces between subsystems are modeled as being
fixed, since they capture hardwired interactions.  
A  key distinction in  multi-agent systems is that the protocols which interface
the network subsystems (agents) are often considered as designable. In particular, the data provided to each agent by the network can
be processed in setting the agent's local input.  The designability of protocols in multi-agent systems can be captured by modeling the subsystem output matrix $C$ 
as being designable rather than fixed.  In this case, the output
matrix $C$ is naturally modeled as taking the form $C=H \widehat{C}$, where $\widehat{C}$ is fixed while $H$ is designable.  This formulation captures that subsystem outputs
defined by $\widehat{C}$ can be processed in setting
local inputs.

{\em Remark:} The model considered here is a specialization
of those considered in \cite{control1,control2,control3,control4}, in that 
the subsystem couplings and external actuation are
both restricted to act through the subsystem input matrix
$B$ (mathematically, the range space of the coupling $BC$
is contained within that of the input matrix $B$).  This additional ``subsystem" structure is apt for many application
domains (see \cite{control2} for further discussion), and has
a rich modal structure that enables structural characterization of controllability.

\section{Spectral Conditions for Controllability}

Conditions for the controllability of the linearly-coupled network model, which are based on a direct eigenvector analysis
of the model, are briefly discussed.  These conditions motivate the main study of repeated eigenvalues in the linearly-coupled network model.  They also represent a correct treatment of controllability of the model, which addresses the omission noted in \cite{comment}.  Since similar conditions were presented in a preliminary form in our previous work \cite{control2} and/or
in the independent study  \cite{control3}, proofs are omitted.
To present the results, it is convenient to define the $N$ eigenvalues of the network matrix $G$ as $\lambda_i$ ($i=1,\cdots,N$). We recall that the $NM$ eigenvalues  of the full state matrix $I\otimes A+G\otimes BC$ are the union of the eigenvalues of $A+\lambda_i BC$ for $i=1,\cdots,N$ \cite{sync1}.

First, a necessary condition for controllability is presented in the following lemma:

\begin{lemma}
The linearly-coupled network, as defined by the pair $(I \otimes A+G \otimes BC, S \otimes B)$, is controllable only if 
1) the pair $(G,S)$ is controllable and 2) the pairs $(A+\lambda_i BC, B)$ are controllable for $i=1,\hdots,N$.  \label{lem:necc} 
\end{lemma}

This necessary condition can be further simplified into
separate conditions on the global and subsystem models, via
a simple application of the Hautus lemma:

\begin{corollary}
The linearly-coupled network, as defined by the pair $(I \otimes A+G \otimes BC, S \otimes B)$, is controllable only if 
1) the global model is controllable (i.e.,  the pair $(G,S)$ is controllable) and 2) the subsystem model is controllable (i.e., the pair $(A, B)$ is controllable).  \label{corr:necc} 
\end{corollary}

There has been inclarity about whether the condition in Lemma \ref{lem:necc} is also sufficient, as noted in \cite{comment,control2,control3}.  Essentially, the reason why the condition is not sufficient is that the matrices $A+\lambda_i BC$, $i=1,\hdots, N$, may share eigenvalues across them. These eigenvalues are repeated eigenvalues of the full network model \eqref{full}, whose eigenspaces may contain eigenvectors other than ones that are Kronecker products of the global- and subsystem-level- eigenvectors.  Hence, it is possible that the PBH test for controllability may fail, even when the global and subsystem models are controllable.  We also notice that the simplification in
Corollary \ref{corr:necc} depends on the network having a ``subsystem" structure, rather than arbitrary couplings among its components.

In the case where the network matrix $G$ is diagonalizable, a
necessary and sufficient condition for controllability can 
be obtained by enumerating the eigenvectors associated with
each repeated eigenvalue of the full network model \eqref{full}.  
The condition requires some additional notation.  We refer to the left eigenvectors of $G$ associated with each eigenvalue $\lambda_j$, $j=1,\hdots, N$,
as ${\bf v}_j^T$.
Further, for each distinct eigenvalue
$\mu_i$, $i=1,\hdots,z$, of the full network dynamics \eqref{full}, 
the geometric multiplicity of the eigenvalue in the matrix $A+\lambda_jBC$ is denoted as $\beta_j(\mu_i)$.  
The corresponding left eigenvectors of $A+\lambda_j BC$ (if there are any)
are denoted as ${\bf w}_{jl}^T$, where $l=1,\hdots, \beta_j(\mu_i)$. Here is the
necessary and sufficient condition:
\begin{lemma}
Assume that the linearly-coupled network \eqref{full} has a diagonalizable network matrix $G$.  
The network, as defined by the pair $(I \otimes A+G \otimes BC, S \otimes B)$, is controllable if and only if 
the following condition holds for each $i=1,\hdots, z$:
$\sum_{j=1}^N \sum_{l=1}^{\beta_{j}(\mu_i)} a_{j,l} ({\bf v}_j^T S) \otimes
({\bf w}_{jl}^T B) \neq 0$
for any real scalar coefficients $a_{j,l}$  that are not all identically $0$.
\label{th:neccsuff}
\end{lemma}

The necessary and sufficient condition in Lemma \ref{th:neccsuff} is important as a complete formal treatment of controllability of the linearly-coupled network \cite{control2,control3}.  However, it is sometimes difficult to apply directly to obtain graph-theoretic or structural results on controllability, because the roles of the global and subsystem models are intertwined in a complicated way.  Noting that the intertwining of the global and subsystem conditions is tied to the presence/absence of repeated eigenvalues across the matrices
$A+ \lambda_i BC$, simpler sufficient or necessary conditions for controllability can be obtained based on eigenvalue multiplicities.  Two useful results of this sort are formalized in the following lemma:
\begin{lemma}
Consider the linearly-coupled network, as defined by the pair $(I \otimes A+G \otimes BC, S \otimes B)$, 
for diagonalizable $G$. The model is controllable if  
the following two conditions hold: 1) the global model $(G,S)$ is controllable; 2) no two matrices $A+\lambda_i BC$ and $A+\lambda_j BC$ ($i=1,\hdots, N$, $j=1,\hdots, N$) such that $\lambda_i \neq \lambda_j$ share eigenvalues. 
Meanwhile, the linearly-coupled network model is not controllable if $Mm < P(\mu)$ for any complex number $\mu$, where $P(\mu)$ is the sum of the geometric
multiplicities of an eigenvalue $\mu$ across matrices $A+\lambda_i BC$.
 \label{lem:suff}
\end{lemma}

The lemma shows that controllability reduces to a global (or graph- based) condition, if the 
matrices $A+\lambda_iBC$ corresponding to distinct $\lambda_i$ do not share eigenvalues. We note here that the subsystem's controllability is necessary (but not sufficient) for the 
matrices $A+\lambda_iBC$ to not share eigenvalues. On the other hand, controllability
is necessarily lost if an eigenvalue recurs in a sufficient
number of blocks $A+\lambda_i BC$ as compared to the number of inputs.  We note that these conditions may not be
tight, as controllability depends on the specific 
structure of the eigenspaces associated with the repeated eigenvalues.  However, the result focuses attention on
the close relationship between repeated eigenvalues
across the matrices $A+\lambda_i BC$ and controllability,
and provides a means to assess controllability solely via an eigenvalue analysis.

{\em Remark:} The controllability analysis become more sophisticated if $G$ is a defective matrix, because the linearly-coupled network's state matrix $I_N \otimes A + G \otimes BC$ may have eigenvectors that neither
Kronecker products of global and subsystem-level eigenvectors, nor combinations thereof.  We refer the reader
to \cite{johnson,xuedefective} for an eigenvector analysis for  Kronecker products of defective matrices,
which is a starting point for a treatment of the linearly-coupled network with
defective $G$.

\section{A Dichotomy of Repeated-Eigenvalue Scenarios}

In this section, we study conditions under which
matrices of the form $A+\lambda BC$ may share eigenvalues for
two distinct complex scalars $\lambda$.  The main outcome
of the analysis is a dichotomy of scenarios: 
\begin{enumerate}
\item  In some cases, matrices $A+ \lambda BC$ have eigenvalues that are fixed over the field of complex $\lambda$, i.e. all matrices of this form have shared eigenvalue(s).
We refer to such eigenvalues as {\em network-invariant modes}. 
\item  Alternately, in other cases, the matrices $A+\lambda BC$ only share eigenvalues if the scalars $\lambda$ are chosen from particular finite sets with bounded cardinality.
We refer to these eigenvalues as {\em special-repeat modes}.
\end{enumerate}
These scenarios are differentiated and characterized in the following development, via formal analyses and examples.  Also, the implications of the scenarios on the controllability of the linear-network model are determined.  A main finding is that network-invariant modes place severe restrictions on controllability of the linearly-coupled network, while 
special-repeat modes cannot modulate controllability.
This analysis also differentiates
controllability of built networks as compared to designable multi-agent systems.  

We note that there has been an intense research effort to characterize the spectra of matrices of the form $A+\lambda BC$ (or $A+ \lambda D$) over the field of complex $\lambda$, primarily to analyze the stability of network synchronization and multi-agent system dynamics \cite{sync1,sync2,sync3,cons1,cons2,cons3}.  The stability analysis is often encoded using a {\em master stability function}, which identifies the region of values $\lambda$ in the complex plane for which Hurwitz stability is achieved.  The possible shapes of the master stability region have been 
studied extensively \cite{chen2}.  To the best of our knowledge, however, scenarios where the matrices $A+\lambda BC$ share eigenvalues at multiple points in the complex plane have not been determined.  Our study approaches this question,  in part by noting connections to
decentralized-control and generalized-eigenvalue notions.

\subsection{Scenarios with Network-Invariant Modes}

We first consider the possibility that the matrices $A+\lambda BC$ may have eigenvalues that remain fixed over the entire field of complex $\lambda$.  
Formally, let us define an eigenvalue that is invariant
over the field of complex $\lambda$ as follows:
\begin{definition}
A complex number $\mu$ is called a {\em network-invariant mode} of the linearly-coupled network \eqref{full}, if $\mu$ is an eigenvalue 
of $A+\lambda BC$ for any complex $\lambda$.  
\end{definition}
We first delineate scenarios under which the linearly-coupled network has network-invariant modes.  We then study the implications of network-invariant modes on the controllability of the linearly-coupled network. 

First, it is apparent that any eigenvalue of $A$ which is
an uncontrollable or unobservable mode of the subsystem model is a network-invariant mode.  This is because an uncontrollable or unobservable eigenvalue is invariant 
to any feedback, in the sense that it is an eigenvalue of $A+BKC$  for any complex matrix $K$, including $K=\lambda I$. 
It is natural to ask whether the uncontrollable and unobservable eigenvalues of the subsystem model are the only network-invariant modes.  If this is the case, then uncontrollability for a generic network matrix may only arise
in the degenerate case that the subsystem model is uncontrollable or unobservable.
However, the following example illustrates that the linearly-coupled network may have network-invariant modes even if the subsystem model is controllable and observable, and indeed controllability
of the linearly-coupled network may be lost:

\begin{example}
Consider a linear-coupled network with: 
$A=\begin{bmatrix} 0 & 0 & 0 \\ 0 & 0 & 1 \\ 1 & 0 & 1 \end{bmatrix}$, $B=\begin{bmatrix} 1 & 0 \\ 0 & 1 \\ 
0 & 0 \end{bmatrix}$, 
$C=\begin{bmatrix} 1 & 0 & 0 \\ 0 & 2 & 0 \end{bmatrix}$,
$G=\begin{bmatrix} 2 & -2 & 0 \\ -2 & 4 & -2 \\ 0 & -2 & 2 \end{bmatrix}$, and $S=\begin{bmatrix} 1 \\ 0 \\ 0 \end{bmatrix}$.  For this model, the local subsystem
is immediately seen to be both observable and controllable.
However, the matrix $A+\lambda BC$ is equal to
$\begin{bmatrix} \lambda & 0 & 0 \\ 0 & 2 \lambda & 1 \\
1 & 0 & 1 \end{bmatrix}$.  Hence, $A+\lambda BC$ has
a network-invariant mode at $\mu=1$.  Noting that
$G$ is diagonalizable, it follows immediately
the state matrix of the linearly-coupled network model has an eigenvalue
$\mu=1$ with geometric multiplicity $3$ (one eigenvalue
at $1$ corresponding to each of the blocks $A+\lambda_iBC$,
where $\lambda_i=0,2,6$ are the eigenvalues of $G$ for the example).  Meanwhile,
the model has only two inputs (both applied
at subsystem 1), hence it cannot be controllable.
\label{ex:1}
\end{example}

Example \ref{ex:1} shows that network-invariant modes may be present even when the local subsystem does not 
have unobservable/uncontrollable modes, because matrices of the form $A+\lambda BC$ represent a structured feedback as compared to arbitrary matrices of the form $A+BKC$.   The example is a starting point point toward a more general understanding of structures that lead to network-invariant modes.  In particular, the local subsystem in the example can be seen to have a {\em decentralized fixed mode} \cite{davison} if each input and measurement
is considered as a separate feedback channel, even though
it is centrally observable and controllable.  Indeed, in general, decentralized fixed modes in the local subsystem are necessarily network-invariant modes of the linearly-coupled network.  To formalize this concept, we apply the definition of a decentralized fixed mode, which was initially developed in the Wang and Davison's seminal work on decentralized control \cite{davison}, to the local-subsystem model.  
\begin{definition}
The complex number $\mu$ is said to be a {\em decentralized fixed mode} of the local-subsystem model $(C,A,B)$, if 
$\mu$ is an eigenvalue of $A+BKC$ for any real diagonal matrix $K$.
\end{definition}

The decentralized fixed modes of the local
subsystem model are a subset of the network-invariant modes, as formalized in the following lemma:
\begin{lemma}
If $\mu$ is a decentralized fixed mode of the local subsystem model, it is also a network-invariant mode of the linearly-coupled network. \label{lem:decent}
\end{lemma}

\begin{proof}
Since $\mu$ is a decentralized fixed mode of the local subsystem model, 
$det(sI-(A+BKC))=0$ has a solution at $s=\mu$ for all real diagonal $K$.  Choosing $K=\lambda I$, we immediately
recover that $det(sI-(A+\lambda BC))=0$ has a solution
at $s=\mu$ for all real $\lambda$, i.e. $det(\mu I-(A+\lambda BC))=0$ for all real $\lambda$.

Now consider $det(\mu I-(A+\lambda BC))$ over the field of complex $\lambda$. The determinant
is a polynomial in $\lambda$ with degree equal to rank
$BC$.  It follows that the equation $det(\mu I-(A+\lambda BC))=0$ either holds for a finite set of $\lambda$ in the complex plane, or holds for all $\lambda$.  Since
the equation holds for all real $\lambda$, it thus must
hold for all $\lambda$.  We have thus shown that
$\mu$ is an eigenvalue of $A+\lambda BC$ for all complex $\lambda$, i.e. it is a network-invariant mode of the linearly-coupled model.  $\square$
\end{proof}

Numerous
algorithmic as well as structural methods for understanding decentralized fixed modes  have been developed in the literature \cite{davison,davison2,morse}.  Per Lemma \ref{lem:decent}, these methods can be applied to characterize local subsystem models
that lead to network-invariant modes.

{\em Remark:} An immediate consequence of Lemma \ref{lem:decent} is that the ordering of the subsystem input and output channels may influence the presence of network-invariant modes, and hence controllability of the network.  For example, in Example $1$, the full network model \eqref{full} becomes controllable when the two columns of $B$ are
switched (i.e., $B=\begin{bmatrix} 0 & 1 \\ 1 & 0 \\ 0 & 0 \end{bmatrix}$).

The linearly-coupled network may be expected to have network-invariant modes other than the subsystem decentralized fixed modes, since the matrix $A+\lambda BC$ has an even more
restricted form than allowed by application of decentralized feedback (specifically, the decentralized feedback gains applied to each channel must be identical).  
The following example shows that the model may have
network-invariant modes even when it does not have 
decentralized fixed modes:

\begin{example}
Consider a linearly-coupled network model with: 
$A=\begin{bmatrix} 0 & 0 & 1 \\ 0 & 0 & 1 \\ 1 & -1 & 1 \end{bmatrix}$, $B=\begin{bmatrix} 1 & 0 \\ 0 & 1 \\ 
0 & 0 \end{bmatrix}$, 
$C=\begin{bmatrix} 1 & 0 & 0 \\ 0 & 1 & 0 \end{bmatrix}$,
 For this model, it can be verified that the local subsystem model
does not have decentralized fixed modes, by checking the eigenvalues of $A+BKC$ for two randomly-selected diagonal $K$
(see \cite{davison}).
However, the matrix $A+\lambda BC$ is equal to
$\begin{bmatrix} \lambda & 0 & 1 \\ 0 & \lambda & 1 \\
1 & -1 & 1 \end{bmatrix}$. For any complex $\lambda$,
this matrix has an eigenvalue at $\mu=1$ with corresponding
right eigenvector equal to $\begin{bmatrix} 1 \\ 1 \\ 1-\lambda
\end{bmatrix}$.  Thus, the linear-network model has a network-invariant 
mode at $\mu=1$ even though the subsystem model
does not have decentralized fixed modes.
\label{ex:2}
\end{example}

The above development has shown that network-invariant modes of the linear network model
are a superset of the decentralized fixed modes of the 
subsystem model (which are themselves a superset of
the unobservable and uncontrollable modes of the subsystem model).  

{\em Remark:} Scaling of
the subsystem inputs or outputs may also influence whether
network-invariant modes are present, and hence may alter the
controllability of the network.  For instance, in Example \ref{ex:2}, if
the $B$ matrix is changed to 
$B=\begin{bmatrix} 1 & 0 \\ 0 & 2 \\ 0 & 0 \end{bmatrix}$,
then the network-invariant mode is eliminated.

{\em Remark:} In the special case that the subsystem model has a single input and single output, the set of network-invariant modes are trivially seen to be precisely the unobservable and uncontrollable modes.

The presence of network-invariant modes necessarily prevents
controllability of the linearly-coupled network, under broad conditions on the global model $(G,S)$.  In particular,  if
a network-invariant mode is present, matrices
$A+\lambda_i BC$ share the mode for all $\lambda_i$ of $G$. Hence, 
this mode has large geometric multiplicity in the full linearly-coupled model \eqref{full} provided
that $G$ is diagonalizable.  Thus, the linearly-coupled network model
can be controllable only if actuation is provided at a sufficient
number of network channels, regardless of the network matrix $G$.  This concept is formalized
in the following lemma:

\begin{lemma}
Consider a linearly-coupled network model \eqref{full} with diagonalizable network
matrix $G$, which has a network-invariant mode. The model is controllable only if the number of actuation locations satisfies $M \ge \lceil \frac{N}{m} \rceil$.
\label{lem:struc}
\end{lemma}

\begin{proof}
For diagonalizable $G$, the geometric multiplicity of
the network-invariant mode is at least $N$.  Thus, 
from Lemma \ref{lem:suff}, controllability requires
that $Mm \ge N$.  The result follows. $\square$
\end{proof}

Lemma \ref{lem:struc} shows that controllability requires actuation at a specified fraction of the subsystems (at least), whenever network-invariant modes are present.  This is true regardless of the graph $\Gamma$ of the linearly-coupled network. 
Controllability {\em cannot} be achieved through design
of the network topology or selection of particular
actuation locations, if the number of actuation locations is insufficient.

Lemma \ref{lem:struc} can be further refined to show that
any weakly-connected partition of the network must have a sufficient
number of actuated nodes for controllability:

\begin{theorem}
Consider a linearly-coupled network with diagonalizable network matrix $G$ which has a network-invariant mode.  Consider any subset ${\cal T}$ of the network's subsystems (correspondingly, vertices in the graph $\Gamma$).  The linearly-coupled network is controllable only
if $\widehat{M}\ge \lceil 
\frac{\widehat{N}}{m}\rceil-b$, where $\widehat{N}$ is the total number of vertices within ${\cal T}$, $\widehat{M}$ is the 
number of actuation locations within ${\cal T}$, and $b$ is the number of vertices
in ${\cal T}$ which  are not actuation locations but have
an incoming edge from vertices outside ${\cal T}$ in $\Gamma$. \label{th:graph}
\end{theorem}

\begin{proof}
Without loss of generality, we assume that ${\cal T}$ contains subsystems $1, \hdots, \widehat{N}$.  The dynamics
of the subsystems within within ${\cal T}$ can be expressed
as
\begin{equation}
\dot{\widehat{\bf x}}=(I_{\widehat{N}} \otimes A
+ \widehat{G} \otimes BC) \widehat{\bf x}
+\widehat{S} \otimes B \widehat{\bf u} . \label{eq:proof}
\end{equation}
In Equation \ref{eq:proof}, $I_{\widehat{N}}$ is an identity
matrix with $\widehat{N}$ rows; $\widehat{G}$ is the principal submatrix of $G$ containing the first $\widehat{N}$ rows and columns; and $\widehat{S}$ is a matrix with $\widehat{N}$ rows and $\widehat{M}+b$ columns, where the columns are 0--1 indicator vector of the $\widehat{M}$ actuation locations and the $z$ additional vertices which
have incoming edges from vertices outside ${\cal T}$.  
The vector $\widehat{\bf u}$ concatenates vectors
$\widehat{\bf u}_i$ corresponding to the $\widehat{M}+z$ 
subsystems indicated by $\widehat{S}$, where each
$\widehat{\bf u}_i$ is the sum of the external input signal
${\bf u}_i$ at the subsystem (if there is an input ${\bf u}_i$ at this subsystem) and a signal projected
from the network outside ${\cal T}$ (specifically,
$\sum_{j \notin {\cal T}} g_{ij} C {\bf x}_j$).  

Controllability of the linearly-coupled network requires controllability of the system \eqref{eq:proof}, where
$\widehat{\bf u}$ is considered as an input signal in \eqref{eq:proof}.  This is because, if system 
\eqref{eq:proof} is not controllable, there is at least
one state $\widehat{\bf x}=\overline{\widehat{ {\bf x}}}$ 
that cannot be achieved via the applied input $\widehat{\bf u}$.  Thus, in this case, the linearly-coupled network model \eqref{full}
also necessarily cannot be driven to an arbitrary
state.  Thus, controllability of \eqref{eq:proof} is necessary for controllability of \eqref{full}.

However, the system \eqref{eq:proof} is a modified linearly-coupled
network model, which has network matrix $\widehat{G}$ rather than $G$, and actuation-location matrix $\widehat{S}$ rather
than $S$.  This modified model is immediately seen to have 
the same network-invariant mode as the original network.  Thus, from Lemma \ref{lem:struc}, 
it follows that $\widehat{M}+z\ge \lceil 
\frac{\widehat{N}}{m}\rceil$, and the theorem statement follows. $\square$
\end{proof}

Theorem \ref{th:graph} shows that 
linearly-coupled networks with network-invariant modes
must have a sufficient density of actuation in all parts of
the network for controllability.  That is, not only is 
actuation required at a certain fraction of the subsystems overall, but each partition of the network requires actuation at a sufficient number of contained subsystems.

For a broad subclass of linear network models with 
network-invariant modes, the requirement for controllability
is even more stringent.  In particular, network-invariant modes often have the characteristic that the eigenvector
of $A+\lambda BC$ associated with the mode has the same
projection on the subsystem input matrix for every $\lambda$.
Formally, let us define $\overline{\bf w}(\lambda)^T$ to be the left
eigenvector of $A+\lambda BC$ associated with a particular
network-invariant mode $\mu$.  We say that the network-invariant mode $\mu$ is {\em projection-fixed}, if the product
$\overline{\bf w}(\lambda)^T B$ is identical for all $\lambda$, to within a multiplicative factor (i.e. $\overline{\bf w}(\lambda)^T B=\alpha({\lambda}){\bf p}$, where ${\bf p}$ is a common row vector for all $\lambda$ and $\alpha({\lambda})$ is a scalar multiplicative factor that depends on $\lambda$).  The following lemma shows
that a linear network model with subsystem-projection-fixed network-invariant modes can only be controllable if 
actuation is provided at all network nodes:

\begin{lemma}
Consider a linearly-coupled network \eqref{full} with diagonalizable $G$.  If the model has any network-invariant modes which are projection-fixed, then the linear network model is controllable only if 
actuation is provided at all network nodes ($M=N$).  
\end{lemma}

\begin{proof}
We apply Lemma \ref{th:neccsuff} to the projection-fixed 
network-invariant mode, say $\mu$, to show necessity.  In particular, the linearly-coupled network is controllable only if: 
$\sum_{j=1}^N \sum_{l=1}^{\beta_{j}(\mu)} \alpha_{j,l} ({\bf v}_j^T S) \otimes
({\bf w}_{jl}^T B) \neq 0$
for any real scalar coefficients $\alpha_{j,l}$  that are not all identically $0$. Since $\mu$ is a network-invariant mode, 
the mode's multiplicity $\beta_j(\mu)$ in each block 
$A+\lambda_j BC$ is necessarily at least $1$. Meanwhile, since
the mode is projection-fixed, ${\bf w}_{jl}^T B$ is
proportional to some row vector ${\bf p}$ for all $j$ and $l$.  
Thus, the necessary condition for controllability can be simplified to
$\sum_{j=1}^N \alpha_j ({\bf v}_j^T S) \otimes {\bf p} 
\neq 0$ for all $\alpha_j$.  However, this is only possible
of $S$ has $N$ independent columns, which necessitates that $M=N$. $\square$
\label{th:projfixed}
\end{proof}

It is easy to check that the network-invariant modes in Examples 1 and 2 are both projection-fixed, and hence actuation at all network locations is necessary for controllability.  Although many network-invariant modes are projection-fixed, not all are.  The following is an example of a linearly-coupled network
with a network-invariant mode that is not projection-fixed:

\begin{example}
Consider a linear network model with: 
$A=\begin{bmatrix} 0.6 & 0.4 & 0 \\ 0.2 & 0.7 & 0.1 \\
0 & 0.2 & 0.8 \end{bmatrix}$, 
$B=\begin{bmatrix}  1 & 0 & 1 & 0 \\ 0 & 1 & 0 &1 \\ 
0 & 0 & 0 & 0 \end{bmatrix}$,
$C=\begin{bmatrix} 1 & 0 & 0 \\ 0 & -1 & 0 \\ 0 & -1 & 0 \\
1 & 0 & 0 \end{bmatrix}$, 
$G=\begin{bmatrix} 1 & -1 & 0 \\ -1 & 2 & -1 \\ 0 & -1 & 1
\end{bmatrix}$, and ${\cal S}=\{ 1, 2 \}$.  
It can be checked that the model has a network-invariant
mode at $\mu =1$ that is not projection-fixed.  The linear
network model is found to be controllable, even though actuation has only been provided at two of the three subsystems.
\end{example}

The above analysis shows that linearly-coupled  networks
may have network-invariant modes, which can place essential limits on controllability regardless of the topology of the network.  Thus, controllability of built networks (i.e., linearly-coupled networks with fixed interfaces) requires additional conditions on the subsystems beyond local controllability and observability. For instance, (centralized) controllability of the linearly-coupled network model 
is tied to decentralized controllability of the subsystem model. 

These barriers to controllability
seem surprising at first glance, since the controllability
of multi-agent system models has been distilled to a purely graph-theoretic condition in the literature \cite{control5}.  This
difference arises because multi-agent system models allow design of the interaction protocols
among the agents, in contrast with built networks.  The difference can be explained within our framework by considering the linearly-coupled network model with designable $C$ matrix, which captures the multi-agent-system context:

\begin{lemma}
Consider a linearly-coupled network model for which the subsystem's output matrix takes the form
$C=H \widehat{C}$, where $H$ is designable.  Provided that
the subsystem model $(\widehat{C}, A, B)$ is observable
and controllable, the matrix $H$ can be designed so that 
the linear network model has no network-invariant modes.
\end{lemma}

\begin{proof}
Consider the eigenvalues of $A+\lambda BC=A+\lambda
B H \widehat{C}$ for $\lambda=0$ and $\lambda=1$.  For
$\lambda=0$, the eigenvalues are those of $A$.  For $\lambda=1$, the eigenvalues are those of $A+BH \widehat{C}$.
Since $(\widehat{C}, A, B)$ is observable
and controllable, $H$ can be chosen so that the eigenvalues
of $A+BH \widehat{C}$ all differ from those of the matrix $A$ (i.e., a feedback can be applied to move all eigenvalues).
For such a choice of $H$, $A+\lambda BC$ does not
share any eigenvalues for $\lambda=0$ and $\lambda=1$, hence
the model has no network-invariant modes.  $\square$
\end{proof}

The lemma formalizes that multi-agent-system protocols can 
always be designed to avoid network-invariant modes, by choosing any matrix $H$ such that the eigenvalues
of $A+B H \widehat{C}$ differ from those of $A$.

\subsection{Scenarios with Special-Repeat Modes}

Even if the linear network model does not have 
network-invariant modes, the blocks $A+\lambda_i BC$ may
share eigenvalues for particular
$\lambda_i$.  These repeated eigenvalues putatively may also
interfere with controllability, since their associated eigenvectors across the blocks may lie in the null space
of the input matrix, per Lemma \ref{th:neccsuff}.
To gain an understanding of whether the blocks $A+\lambda_i BC$ can share eigenvalues for special choices of $\lambda_i$, we consider solving for the $\lambda$ for which $A+\lambda BC$ has a particular eigenvalue $\mu$.  For such $\lambda$, 
the equation ${\bf w}^T (A+\lambda BC) =\mu {\bf w}^T$ must
be satisfied for some ${\bf w}^T$.  Rearranging, the
equation can be written as:
\begin{equation}
{\bf w}^T (A- \mu I) =\lambda {\bf w}^T (-BC) \label{eq:geneig}
\end{equation}  

Equation \ref{eq:geneig} clarifies that finding the
$\lambda$ values that yield a particular eigenvalue $\mu$
corresponds to solving a {\em generalized eigenvalue problem} for the pair $(A- \mu I, -BC)$.  From standard results on generalized eigenvalue
problems \cite{geneig1,geneig2}, only the following scenarios are
possible:
\begin{itemize}
\item  The generalized eigenvalue problem is degenerate, in the sense
that it has a solution for every $\lambda$ (i.e., there 
is vector in the left-null-space of $A+\lambda BC-\mu I$
from every $\lambda$). This corresponds to the case that
 that the linear network model has network-invariant modes.
\item  Alternately, the generalized eigenvalue problem
has a solution for a finite set of $\lambda$, with the number of distinct solutions less than or equal to the rank of the matrix $BC$.  We refer to eigenvalues $\mu$ that only
repeat for a finite set of $\lambda$ as special-repeat 
modes, and refer to the corresponding set of distinct $\lambda$ values
as the {\em network-repeat set for $\mu$ (or $NR(\mu)$}).   
\end{itemize}

The generalized-eigenvalue formulation leads
to the following characterization of the
special-repeat modes and network-repeat sets:
\begin{lemma}
Consider an eigenvalue $\mu$ that is a special-repeat mode of a linearly-coupled network \eqref{full}.
The maximum number of distinct $\lambda$ for which 
$A+\lambda BC$ has eigenvalue $\mu$ (i.e., the maximum possible size for the network-repeat set $NR(\mu)$) is at most $rank(BC)$.  Further,
consider the set of vectors ${\bf w}^T$ for which ${\bf w}^T (A+\lambda BC) =\mu {\bf w}^T$, for $\lambda \in NR(\mu)$.  There are $r \le rank(BC)$ such vectors, say
${\bf w}_1^T,\hdots, {\bf w}_r^T$, which are linearly independent.
Further, the projections of these vectors into the matrix $B$ are
linearly independent, i.e. ${\bf w}_1^T B,\hdots, {\bf w}_r^T B$
are linearly independent.
\label{lem:repeat}
\end{lemma}

\begin{proof}
The number of the solutions $\lambda$, and the fact that there are
at most $rank(BC)$ vectors that
satisfy ${\bf w}^T (A+\lambda BC) =\mu {\bf w}^T$, are standard results 
on generalized eigenvalues \cite{geneig1}.  

It remains to prove
that the projections of the generalized eigenvectors ${\bf w}_1^T B,\hdots, {\bf w}_r^T B$ are linearly independent.  
We will prove the result by contradiction.  If the vectors are not
linearly independent, then there exist $\gamma_1,\hdots, \gamma_r$ which
are not identically zero such that $\sum_{i=1}^r \gamma_i {\bf w}_i^T B=0$.
Notice that this is only possible for $r \ge 2$. To continue, we scale and sum
the equations ${\bf w}_i^T(A+\lambda_i BC)=\mu {\bf w}_i^T$ for $i=1,\hdots, r$, where
$\lambda_i$ are the corresponding generalized eigenvalues in $NR(\mu)$.  In particular, we
can get:
\begin{equation}
\sum_{i=1}^r \frac{\gamma_i+c_i}{\lambda_i}{\bf w}_i^T
(A+\lambda_i BC)=\mu \sum_{i=1}^r \frac{\gamma_i+c_i}{\lambda_i}{\bf w}_i^T, \label{eq:lin}
\end{equation}
where $c_1,\hdots,c_r$ are scalars.  Multiplying out the left
side of \eqref{eq:lin} and then using $\sum_{i=1}^r \gamma_i {\bf w}_i^T B=0$, we 
get:
\begin{equation}
\sum_{i=1}^r \frac{\gamma_i+c_i}{\lambda_i}{\bf w}_i^T A+
\sum_{i=1}^r c_i {\bf w}_i^T BC=\mu \sum_{i=1}^r \frac{\gamma_i+c_i}{\lambda_i}{\bf w}_i^T, 
\end{equation}
Choosing $c_i=\frac{h \gamma_i}{\lambda_i-h}$ for any scalar $h \neq \lambda_i$, one recovers
that:
\begin{equation}
\sum_{i=1}^r c_i{\bf w}_i^T (A+hBC) = \mu \sum_{i=1}^r c_i{\bf w}_i^T .
\end{equation}
Thus, $\mu$ is seen to be an eigenvalue of $A+hBC$ for any scalar $h$, with 
corresponding eigenvector $\sum_{i=1}^r c_i{\bf w}_i^T$.  This contradicts the fact that 
$\mu$ is a special-repeat mode, since it is an eigenvalue for a continuum of $h$ rather than a finite set.
Thus, the result is proved by contradiction. $\square$
\end{proof}

Lemma \ref{lem:repeat} indicates the matrices $A+\lambda_i BC$ may have repeated eigenvalues across them, however the total multiplicity is limited by $rank(BC)$,
and further the corresponding eigenvectors' projections on $B$ are linearly independent.  In fact, these limitations on 
special-repeat modes' eigenspaces have implications on controllability of the linearly-coupled network model, as is formalized in the following
theorem:
\begin{theorem}
Consider a linearly-coupled network \eqref{full} with diagonalizable $G$, Further, assume that all eigenvalues 
of the model are special-repeat modes (i.e., the model has no network-invariant
modes). Then, if the global model is controllable, the 
linearly-coupled network model is controllable.
\label{th:special}
\end{theorem}
\begin{proof}
From Lemma \ref{th:neccsuff}, the linearly-coupled network model is controllable if 
$\sum_{j=1}^N \sum_{l=1}^{\beta_{j}(\mu)} \alpha_{j,l} ({\bf v}_j^T S) \otimes
({\bf w}_{jl}^T B) \neq 0$ for all $\alpha_{j,l}$ which are not identically $0$, 
for each distinct eigenvalue $\mu$ of the linearly-coupled network model.  
Consider a particular eigenvalue $\mu$, which by assumption is a special-repeat mode. 
Consider the case that $\mu$ is an eigenvalue of 
$A+\lambda_i BC$ for $g$ distinct eigenvalues $\lambda_i$ of $G$, where 
$g\le rank(BC)$ from Lemma \ref{lem:repeat}; let us refer to these eigenvalues of $G$
a $\lambda_1,\hdots,\lambda_g$ without loss of generality.  (We notice that these
eigenvalues of $G$ also may also recur in other blocks $A+\lambda_i BC$).  
In this notation, the controllability condition can be written as:
$\sum_{j=1}^g \sum_{l=1}^{\beta_j(\mu)} \alpha_{j,l} ({\bf v}_j^T S) \otimes ({\bf w}_{jl}^T B)$,
where ${\bf v}_j^T$ may be any eigenvector associated with the eigenvalue $\lambda_j$.
From Lemma \ref{lem:repeat}, we notice that ${\bf w}_{jl}^T B$ for each pair $(j,l)$ is linearly
independent.  Thus, the eigenvalue $\mu$ is controllable if ${\bf v}_j^T S \neq 0$ for 
all $j$.  However, this is guaranteed from controllability of the global model.  
The argument can be repeated for each distinct eigenvalue $\mu$ of the linearly-coupled
network model, hence controllability is shown. $\square$
\end{proof}

Theorem \ref{th:special} demonstrates that special-repeat eigenvalues are harmless, in the sense that they cannot prevent controllability of the linear network model.  

{\em Remark:} In contrast with the result presented here, special-repeat-type modes (i.e., eigenvalues of $A+\lambda BC$ that repeat only for finite sets of $\lambda$) can 
block controllability for the broader class of network
models considered in \cite{control1,control2,control3}.  Indeed, the counterexample shown in the comment \cite{comment} arises from
a special-repeat-type mode.  The reason for the difference is
that the linear couplings and the subsystem input matrix
do not have the same range space in these models (i.e., the coupling matrix is $D \neq BC$).
In this broader circumstance, special-repeat modes can result in uncontrollability, which means that 
particular choices of the network matrix may lead to uncontrollability while other choices allow control.

\subsection{Discussion}

Our analysis has shown that the eigenvalues of the linear
network model can be categorized into two types, 
which we call network-invariant modes and special-repeat modes.  The two types of modes, and their implications on network controllability, have been characterized above.  A key outcome of the analysis is that controllability of the linear-network model is solely tied to the presence or absence of network-invariant modes, rather than special-repeat modes.  This characterization of
controllability has several interesting implications:
\begin{itemize}
\item  A linearly-coupled network model without network-invariant modes (with diagonalizable $G$) is  controllable only if the global model $(S,B)$ is controllable.  Thus, controllability
can be distilled to a condition on only the network 
topology and input locations in this case.  The wide
range of graph-theoretic results on the controllability
of scalar diffusive/consensus processes can be brought to
bear (e.g. \cite{egerstedt}), if the topology matrix $G$ has a diffusive structure.
\item  If the linearly-coupled network model has network-invariant modes, then controllability is lost regardless of the network's graph topology, unless actuation is provided
at a sufficient fraction of nodes.  If these 
network-invariant modes are projection-fixed, then 
controllability is lost unless actuation is provided
at all nodes in the network.
\item  If the subsystem model is SISO, then the linearly-coupled network model has network-invariant modes if and only if the 
subsystem is either unobservable or uncontrollable. 
Thus, the linear network model is controllable if 
1) the global model is controllable, and 2) the subsystem model is controllable and observable.
\item  The controllability of built networks differs from that of multi-agent systems, because of the design freedom
available in multi-agent-system control.  In particular,
protocols in multi-agent systems can be designed to 
avoid network-invariant modes, which means that 
they can be designed so that controllability reduces entirely
to a graph-level condition.  Indeed, the study \cite{control5} 
provides a design process that ensures controllability, in a state-feedback setting.
\item  While the effort here has been focused on controllability, the results naturally translate to the dual question of observability from measurements of a subset of subsystems. 
\end{itemize}

\bibliographystyle{IEEEtran}
\bibliography{modalbib}

\begin{thebibliography}{10}
\providecommand{\url}[1]{#1}
\csname url@samestyle\endcsname
\providecommand{\newblock}{\relax}
\providecommand{\bibinfo}[2]{#2}
\providecommand{\BIBentrySTDinterwordspacing}{\spaceskip=0pt\relax}
\providecommand{\BIBentryALTinterwordstretchfactor}{4}
\providecommand{\BIBentryALTinterwordspacing}{\spaceskip=\fontdimen2\font plus
\BIBentryALTinterwordstretchfactor\fontdimen3\font minus
  \fontdimen4\font\relax}
\providecommand{\BIBforeignlanguage}[2]{{%
\expandafter\ifx\csname l@#1\endcsname\relax
\typeout{** WARNING: IEEEtran.bst: No hyphenation pattern has been}%
\typeout{** loaded for the language `#1'. Using the pattern for}%
\typeout{** the default language instead.}%
\else
\language=\csname l@#1\endcsname
\fi
#2}}
\providecommand{\BIBdecl}{\relax}
\BIBdecl

\bibitem{sync1}
C.~W. Wu and L.~O. Chua, ``Synchronization in an array of linearly coupled
  dynamical systems,'' \emph{IEEE Transactions on Circuits and Systems I:
  Fundamental Theory and Applications}, vol.~42, no.~8, pp. 430--447, 1995.

\bibitem{sync2}
X.~F. Wang and G.~Chen, ``Synchronization in scale-free dynamical networks:
  robustness and fragility,'' \emph{IEEE Transactions on Circuits and Systems
  I: Fundamental Theory and Applications}, vol.~49, no.~1, pp. 54--62, 2002.

\bibitem{sync3}
L.~M. Pecora and T.~L. Carroll, ``Master stability functions for synchronized
  coupled systems,'' \emph{Physical review letters}, vol.~80, no.~10, p. 2109,
  1998.

\bibitem{cons1}
R.~Olfati-Saber, J.~A. Fax, and R.~M. Murray, ``Consensus and cooperation in
  networked multi-agent systems,'' \emph{Proceedings of the IEEE}, vol.~95,
  no.~1, pp. 215--233, 2007.

\bibitem{cons2}
T.~Yang, S.~Roy, Y.~Wan, and A.~Saberi, ``Constructing consensus controllers
  for networks with identical general linear agents,'' \emph{International
  Journal of Robust and Nonlinear Control}, vol.~21, no.~11, pp. 1237--1256,
  2011.

\bibitem{cons3}
Z.~Li, W.~Ren, X.~Liu, and M.~Fu, ``Consensus of multi-agent systems with
  general linear and lipschitz nonlinear dynamics using distributed adaptive
  protocols,'' \emph{IEEE Transactions on Automatic Control}, vol.~58, no.~7,
  pp. 1786--1791, 2013.

\bibitem{chen1}
Z.~Li, Z.~Duan, G.~Chen, and L.~Huang, ``Consensus of multiagent systems and
  synchronization of complex networks: A unified viewpoint,'' \emph{IEEE
  Transactions on Circuits and Systems I: Regular Papers}, vol.~57, no.~1, pp.
  213--224, 2010.

\bibitem{control1}
S.~Zhang, M.~Cao, and M.~K. Camlibel, ``Upper and lower bounds for controllable
  subspaces of networks of diffusively coupled agents,'' \emph{IEEE
  Transactions on Automatic control}, vol.~59, no.~3, pp. 745--750, 2014.

\bibitem{control2}
M.~Xue and S.~Roy, ``Input-output properties of linearly-coupled dynamical
  systems: Interplay between local dynamics and network interactions,'' in
  \emph{Decision and Control, 2017 IEEE 56th Annual Conference on}.\hskip 1em
  plus 0.5em minus 0.4em\relax IEEE, 2017, pp. 487--492.

\bibitem{control3}
Y.~Hao, Z.~Duan, and G.~Chen, ``Further on the controllability of networked
  mimo lti systems,'' \emph{International Journal of Robust and Nonlinear
  Control}, vol.~28, no.~5, pp. 1778--1788, 2018.

\bibitem{control4}
L.~Wang, G.~Chen, X.~Wang, and W.~K. Tang, ``Controllability of networked mimo
  systems,'' \emph{Automatica}, vol.~69, pp. 405--409, 2016.

\bibitem{control5}
Z.~Ji, H.~Lin, and H.~Yu, ``Protocols design and uncontrollable topologies
  construction for multi-agent networks,'' \emph{IEEE Transactions on Automatic
  Control}, vol.~60, no.~3, pp. 781--786, 2015.

\bibitem{miss1}
N.~Cai and Y.-S. Zhong, ``Formation controllability of high-order linear
  time-invariant swarm systems,'' \emph{IET control theory \& applications},
  vol.~4, no.~4, pp. 646--654, 2010.

\bibitem{comment}
M.~Xue and S.~Roy, ``Comment on ``upper and lower bounds for controllable
  subspaces of networks of diffusively-coupled agents”,'' \emph{IEEE
  Transactions on Automatic Control}, 2017.

\bibitem{johnson}
R.~A. Horn and C.~R. Johnson, \emph{Matrix analysis}.\hskip 1em plus 0.5em
  minus 0.4em\relax Cambridge university press, 1990.

\bibitem{xuedefective}
M.~Xue and S.~Roy, ``Kronecker products of defective matrices: Some spectral
  properties and their implications on observability,'' in \emph{American
  Control Conference, 2012}.\hskip 1em plus 0.5em minus 0.4em\relax IEEE, 2012,
  pp. 5202--5207.

\bibitem{chen2}
Z.~Duan, G.~Chen, and L.~Huang, ``Synchronization of weighted networks and
  complex synchronized regions,'' \emph{Physics Letters A}, vol. 372, no.~21,
  pp. 3741--3751, 2008.

\bibitem{davison}
S.-H. Wang and E.~Davison, ``On the stabilization of decentralized control
  systems,'' \emph{IEEE Transactions on Automatic Control}, vol.~18, no.~5, pp.
  473--478, 1973.

\bibitem{davison2}
E.~Davison and S.~Wang, ``A characterization of decentralized fixed modes in
  terms of transmission zeros,'' \emph{IEEE Transactions on Automatic Control},
  vol.~30, no.~1, pp. 81--82, 1985.

\bibitem{morse}
B.~D. Anderson and D.~J. Clements, ``Algebraic characterization of fixed modes
  in decentralized control,'' \emph{Automatica}, vol.~17, no.~5, pp. 703--712,
  1981.

\bibitem{geneig1}
D.~S. Watkins, \emph{The matrix eigenvalue problem: GR and Krylov subspace
  methods}.\hskip 1em plus 0.5em minus 0.4em\relax Siam, 2007, vol. 101.

\bibitem{geneig2}
L.~Qiu and E.~J. Davison, ``The stability robustness of generalized
  eigenvalues,'' in \emph{Decision and Control, 1989., Proceedings of the 28th
  IEEE Conference on}.\hskip 1em plus 0.5em minus 0.4em\relax IEEE, 1989, pp.
  1902--1907.

\bibitem{egerstedt}
A.~Rahmani, M.~Ji, M.~Mesbahi, and M.~Egerstedt, ``Controllability of
  multi-agent systems from a graph-theoretic perspective,'' \emph{SIAM Journal
  on Control and Optimization}, vol.~48, no.~1, pp. 162--186, 2009.

\end{thebibliography}

\end{document}